\documentclass{article}
\usepackage{amsmath}
\usepackage{parskip}
\usepackage[hidelinks]{hyperref}
\usepackage{graphicx}
\usepackage{natbib}
\usepackage[bf]{caption}
\usepackage{amsthm}
\newtheorem{proposition}{Proposition}[section]
\newtheorem{lemma}{Lemma}[section]
\usepackage{subcaption}

\usepackage{comment}
\usepackage{color}
\title{On effects of inhomogeneity on anisotropy in Backus average}
\author{
Filip P. Adamus%
\footnote{
Department of Earth Sciences, Memorial University of Newfoundland, Canada, {\tt adamusfp@gmail.com}}\,,
Ayiaz Kaderali%
\footnote{
Department of Earth Sciences, Memorial University of Newfoundland, Canada, {\tt ayiazkaderali@gmail.com}}\,,
Michael A. Slawinski%
\footnote{
Department of Earth Sciences, Memorial University of Newfoundland, Canada, {\tt mslawins@mac.com}}\,,
Theodore Stanoev%
\footnote{
Department of Earth Sciences, Memorial University of Newfoundland, Canada, {\tt theodore.stanoev@gmail.com}}
}
\date{}

\begin{document}
\maketitle
%%%%%%%%%%%%%%%%%%%%%%%%%
\begin{abstract}
In general, the Backus average of an inhomogeneous stack of isotropic layers is a transversely isotropic medium.
Herein, we examine a relation between this inhomogeneity and the strength of resulting anisotropy, and show that, in general, they are proportional to one another.
There is an important case, however, in which the Backus average of isotropic layers results in an isotropic---as opposed to a transversely isotropic---medium.
We show that it is a consequence of the same rigidity of layers, regardless of their compressibility.
Thus, in general, the strength of anisotropy of the Backus average increases with the degree of inhomogeneity among layers, except for the case in which all layers exhibit the same rigidity.
\end{abstract}
%%%%%%%%%%%%%%%%%%%%%%%%%
\section{Introduction}
%%%%%%%%%%%%%%%%%%%%%%%%%
\subsection{Backus average}
%%%%%%%%%%%%%%%%%%%%%%%%%
In this paper, we discuss the~\citet{Backus1962} average of isotropic layers as a measure of inhomogeneity of these layers.
Herein, the~\citeauthor{Backus1962} average results in a homogeneous transversely isotropic medium.
Each isotropic layer is defined by the elasticity parameters, $c_{1111} := c^{\ast}_{1111}/\rho = v_P^2$ and $c_{2323} := c^{\ast}_{2323}/\rho = v_S^2$\,, whose values we calculate using the $P$- and $S$-wave speeds, $v_P$ and $v_S$\,, which are obtained from compressional- and shear-sonic logs.
Throughout this paper, we use $c_{1111}$ and $c_{2323}$\,, which are the elasticity parameters scaled by density,~$\rho$\,, as opposed to their non-scaled counterparts, $c^{\ast}_{1111}$ and $c^{\ast}_{2323}$\,.
The corresponding five parameters of the transversely isotropic medium are
\begin{equation}
\label{eq:Tue1}
c^{\overline{\rm
TI}}_{1111}=\overline{\left(\dfrac{c_{1111}-2c_{2323}}{c_{1111}}\right)}^{\,2}
\,\,\,\overline{\left(\dfrac{1}{c_{1111}}\right)}^{\,-1}
+\overline{\left(\dfrac{4(c_{1111}-c_{2323})c_{2323}}{c_{1111}}\right)}\,,
\end{equation}
\begin{equation}
\label{eq:Backus1133}
c^{\overline{\rm TI}}_{1133}=\overline{\left(\dfrac{c_{1111}-2c_{2323}}{c_{1111}}\right)}\,\,
\,\,\overline{\left(\dfrac{1}{c_{1111}}\right)}^{\,-1}
\,,
\end{equation}
\begin{equation}
\label{eq:Berry1}
c^{\overline{\rm TI}}_{1212}=\overline{c_{2323}}
\,,
\end{equation}
\begin{equation}
\label{eq:Berry2}
c^{\overline{\rm TI}}_{2323}=\overline{\left(\dfrac{1}{c_{2323}}\right)}^{\,-1}
\,,
\end{equation}
\begin{equation}
\label{eq:Tue2}
c^{\overline{\rm TI}}_{3333}=\overline{\left(\dfrac{1}{c_{1111}}\right)}^{\,-1}
\,.
\end{equation}
Herein, the bar indicates an average, which is defined by~\citet{Backus1962} as
\begin{equation}
\label{eq:bacintegral}
 \overline{f}(x_3)=\int\limits_{-\infty}^\infty w(\xi-x_3)f(\xi)\,{\rm d}\xi
 \,,
\end{equation}
where the weight,~$w(x_3)$\,, allows us the use of many functions, since the conditions imposed on it are not restrictive.
$w$ is required to be a continuous nonnegative function tending to zero at infinities and to exhibit the following properties:
\begin{equation*}
\label{eq:wInt1}
\int\limits_{-\infty}^\infty w(x_3)\,{\rm d}x_3=1
 \,,
\end{equation*}
\begin{equation*}
\int\limits_{-\infty}^\infty x_3\,w(x_3)\,{\rm d}x_3=0
\qquad{\rm and}\qquad
\int\limits_{-\infty}^\infty x_3^2\,w(x_3)\,{\rm d}x_3=\ell'^{\,2}
\,,
\end{equation*}
where $\ell'$ denotes the width of the stack of parallel layers.
In this paper, for computational purposes, we assume equal thicknesses of the layers, where the thickness weights the average.
Also, we assume that the stack of layers stands for the interval of the average; therefore, we use an arithmetic average.
Readers interested in further details of the~\citeauthor{Backus1962} average might refer to~\citet{BosEtAl2017,BosEtAl2018}.
%%%%%%%%%%%%%%%%%%%%%%%%%
\subsection{Thomsen parameters}
%%%%%%%%%%%%%%%%%%%%%%%%%
To examine the strength of anisotropy of a transversely isotropic homogeneous medium, we invoke~\citet{Thomsen1986} parameters,
\begin{equation}\label{eq:gamma}
\gamma
:=
\dfrac{
c^{\overline{\rm TI}}_{1212} - c^{\overline{\rm TI}}_{2323}
}{
2\,c^{\overline{\rm TI}}_{2323}
}
\,,
\end{equation}
\begin{equation}\label{eq:delta}
\delta
:=
\dfrac{
\left(c^{\overline{\rm TI}}_{1133}+c^{\overline{\rm TI}}_{2323}\right)^{\!2} -
\left(c^{\overline{\rm TI}}_{3333}-c^{\overline{\rm TI}}_{2323}\right)^{\!2}
}{
2\,c^{\overline{\rm TI}}_{3333}\left(
c^{\overline{\rm TI}}_{3333} - c^{\overline{\rm TI}}_{2323}
\right)
}
\,,
\end{equation}
\begin{equation}\label{eq:epsilon}
\epsilon
:=
\dfrac{
c^{\overline{\rm TI}}_{1111}-c^{\overline{\rm TI}}_{3333}
}{
2\,c^{\overline{\rm TI}}_{3333}
}
\,.
\end{equation}
A quantitative measure on the strength of anisotropy is given by the absolute values of these parameters. 
In the case of isotropy, they are zero.
%%%%%%%%%%%%%%%%%%%%%%%%%
\section{Effects of inhomogeneity on anisotropy}
%%%%%%%%%%%%%%%%%%%%%%%%%
\subsection{Alternating layers: Anisotropic medium}
%%%%%%%%%%%%%%%%%%%%%%%%%
In the context of the~\citeauthor{Backus1962} average,~\citeauthor{Thomsen1986} parameters can be also used to infer the effects of inhomogeneity between layers.
In general, as the inhomogeneity within a stack of layers increases, so does the anisotropy of the medium.

To exemplify this increase, let us consider a stack of identical isotropic layers.
To introduce inhomogeneity, we multiply the two elasticity parameters of every second layer by $a$\,; we obtain $c_{1111}$\,, $c_{2323}$ and $a\,c_{1111}$\,, $a\,c_{2323}$\,, for the adjacent layers.
Using, for such a model, expressions~(\ref{eq:Tue1})--(\ref{eq:Tue2}), we obtain the parameters of a transversely isotropic medium, $c_{ijk\ell}^{\overline{\rm TI}}$\,, which, in turn, we use in expressions~(\ref{eq:gamma})--(\ref{eq:epsilon}) to obtain 
\begin{equation}\label{eq:BackusGamma}
\gamma = \dfrac{\left(a-1\right)^2}{8\,a}\,,
\end{equation}
\begin{equation*}
\delta = 0\,,
\end{equation*}
\begin{equation}\label{eq:BackusEpsilon}
\epsilon 
= 
\dfrac{
\left(a-1\right)^2\left(c_{1111}-c_{2323}\right)c_{2323}
}{
2\,a\,c_{1111}^2
}
\,.
\end{equation}
In contrast to parameters~(\ref{eq:gamma}) and (\ref{eq:epsilon}), in general, their counterparts~(\ref{eq:BackusGamma}) and (\ref{eq:BackusEpsilon}), for this model, can be only nonnegative.
Also, $\delta=0$ is a consequence of alternating layers whose both parameters are scaled by the same value of~$a$\,; it is not a general property for alternating isotropic layers in the context of the~\citeauthor{Backus1962} average.

If $a=1$\,, which means that all layers are the same, then also $\gamma=\epsilon=0$\,; hence, in such a case, the averaged medium is isotropic, as expected.
If $a\to0$ or $a\to\infty$\,, which is tantamount to increasing inhomogeneity between layers, then $\gamma$ and $\epsilon$ tend to infinity; in such a case, the averaged medium is extremely anisotropic.

%%%%%%%%%%%%%%%%%%%%%%%%%
\begin{figure}[h]
\centering
\includegraphics[width=8cm]{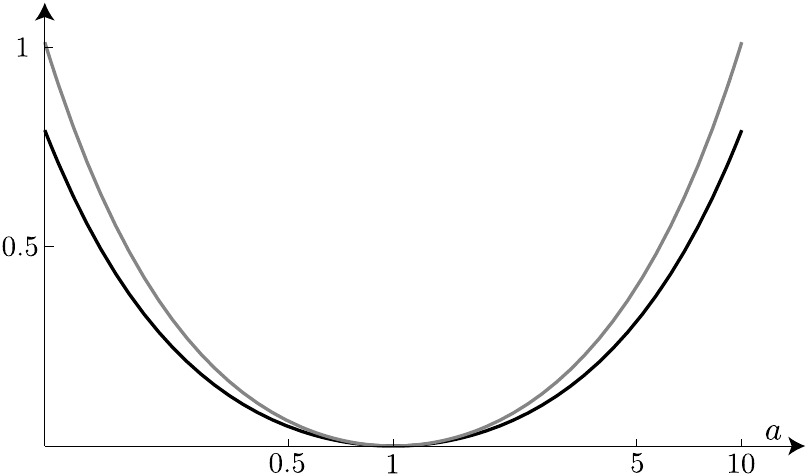}
%\caption[Backus average as function of inhomogeneity]
\caption{\small{Anisotropy of the~\citeauthor{Backus1962} average as a function of layer inhomogeneity:~\citeauthor{Thomsen1986} parameters\,, $\gamma$ and $\epsilon$\,, plotted as grey and black lines, respectively, against logarithmic values of~$a\in\left(10^{-1},10^1\right)$\,.}}
\label{fig:thomsen}
\end{figure}
%%%%%%%%%%%%%%%%%%%%%%%%%

To illustrate the relationship between inhomogeneity and anisotropy, let us consider a numerical example. 
We use $c_{1111} = 12.15$ and $c_{2323} = 3.24$\,, which are density-scaled elasticity parameters that correspond to sandstone.
Their $SI$ units are $\rm km^2/s^2$\,, and their square roots are $P$-wave and $S$-wave speeds, respectively.
Figure~\ref{fig:thomsen} illustrates a monotonic increase in anisotropy of the averaged medium with an increase of inhomogeneity between layers.
At $a=1$\,, which means that all layers are the same, $\gamma=\epsilon=0$\,.
As~$a$ tends to zero or to infinity, $\gamma$ and $\epsilon$ tend to infinity.
For $a\in(10^{-1},10^0)$\,, the values of the elasticity parameters of the alternating layer are progressively diminished by up to one order of magnitude; for $a\in(10^0,10^1)$\,, they are progressively increased by up to one order. 

For the $SH$ and $qP$ waves, respectively, $\gamma$ and $\epsilon$ are measures of difference between propagation speeds along, and perpendicular to, the layers,
\begin{equation*}
\dfrac{v^{2}_{\parallel} - v^{2}_{\perp}}{2\,v^{2}_{\perp}}
\,.
\end{equation*}
Parameter~$\delta$\,, whose definition does not have such a geometrical interpretation, remains equal to zero.
If, however, the elasticity parameters of the alternate layers are $a\,c_{1111}$ and $\sqrt a\,c_{2323}$\,, $\delta$ asymptotically approaches a finite value, as~$a$~tends to infinity; $\gamma$ and $\epsilon$ still tend to infinity and, as such, they are symptomatic of inhomogeneity among layers.

As illustrated in Figure~\ref{fig:thomsen}, for a stack of isotropic layers, the strength of anisotropy of the resulting transversely isotropic medium is solely a function of inhomogeneity of that stack.
In other words, herein, the strength of anisotropy is a measure of inhomogeneity.

A rather slow increase of values of $\gamma$ and $\epsilon$ as functions of~$a$ supports the adequacy of weakly anisotropic models in many quantitative studies in seismology.
Herein, according to the~\citeauthor{Backus1962} average, even moderately inhomogeneous alternating layers result only in a weakly anisotropic medium.
%%%%%%%%%%%%%%%%%%%%%%%%%
\subsection{Isotropic layers: Isotropic medium} \label{sec:iso}
%%%%%%%%%%%%%%%%%%%%%%%%%
Even though, in general, isotropic layers result---by the~\citeauthor{Backus1962} average---in a transversely isotropic medium, there exists a case for which inhomogeneity of the stack of isotropic layers results in an isotropic medium.
In such a case, the inhomogeneity among layers is expressed only by differences in $c_{1111}$\,; $c_{2323}$ remains constant.
\citet[Section~6]{Backus1962} states that
\begin{quote}
if a layered isotropic medium has constant $\mu$\,, the STILWE medium is isotropic.%
\footnote{In this quote, $\mu\equiv c_{2323}$ and STILWE stands for smoothed, transversely isotropic, long-wave equivalent.}
This much was proved by \citet{Postma1955} for periodic two-layered media.
\end{quote}

Let us examine such a case.
Following expressions~(\ref{eq:Tue1})--(\ref{eq:Tue2}), and using a symbolic-calculation software---without any assumption of periodicity \citep[p.~788]{Postma1955}---we obtain,
\begin{align}\label{eq:spec1}
c^{\overline{\rm TI}}_{1111}
&=
\overline{\left(\dfrac{1}{c_{1111}}\right)}^{\,-1}
\,,
\\
c^{\overline{\rm TI}}_{1133}
&=
\overline{\left(\dfrac{1}{c_{1111}}\right)}^{\,-1}-2c_{2323}
\,,
\\
c^{\overline{\rm TI}}_{1212} 
&= 
c_{2323}
\,,
\\
c^{\overline{\rm TI}}_{2323} 
&= 
c_{2323}
\,,
\\
c^{\overline{\rm TI}}_{3333}\label{eq:spec5}
&=
\overline{\left(\dfrac{1}{c_{1111}}\right)}^{\,-1}
\,,
\end{align}
respectively.
Since 
$c^{\overline{\rm TI}}_{1111} = c^{\overline{\rm TI}}_{3333}$\,, 
$c^{\overline{\rm TI}}_{1212} = c^{\overline{\rm TI}}_{2323}$ and 
$c^{\overline{\rm TI}}_{1133} = c^{\overline{\rm TI}}_{1111} - 2\,c^{\overline{\rm TI}}_{2323}$\,,
the medium is isotropic.

In view of the mechanical interpretation of $c_{1111}$ and $c_{2323}$~\citep[e.g.,][Section~5.12.4]{Slawinski2020a}, expressed in terms of the Lam\'e parameters, this result shows that the anisotropy of the~\citeauthor{Backus1962} average is not a consequence of inhomogeneity, in general, but of the difference in the rigidity among the layers.
The difference in compressibility alone does not result in an anisotropic medium.

In terms of wave propagation, the speed of a shear wave, $v^{2}_{S} = c^{\overline{\rm TI}}_{2323} = c_{2323}$\,, depends on rigidity, which is constant, and the speed of a pressure wave, $v^{2}_{P} = c^{\overline{\rm TI}}_{1111}$\,, on the average compressibility.
Since, as shown by~\citet{Rochester2010}, in the context of the necessary and sufficient conditions, the shear wave is due to an equivoluminal deformation, $\nabla\times u$\,, and the pressure wave is due to dilatation, $\nabla\cdot u$\,, where $u$ stands for displacement, it is reasonable to expect anisotropy to originate in a vectorial, not a scalar, quantity.

Let us exemplify such a case with field data, from well-logging measurements offshore Newfoundland. 
We consider a small portion of the data, where each measurement interval is a thin layer within the shale unit.
Sonic logs are used to obtain the $P$- and $S$-wave speeds; the gamma ray log is used to confirm the lithology. 
In lieu of using density logs, whose reliability is questionable, we consider density-scaled elasticity parameters.

The~\citeauthor{Backus1962} average, means and standard deviations of the elasticity parameters, and~\citeauthor{Thomsen1986} parameters are shown in Table~\ref{tab:table1ass}.
Examining the standard deviation in Table~\ref{tab:table1ass}(b), we infer---from their relatively small values---that the properties of layers vary little. 
This small difference can be attributed to the thickness of the overburden increasing with depth as well as to a slightly different composition.
Comparing the standard deviations of $c_{1111}$ and $c_{2323}$\,, we confirm that the former varies more than the latter; the ten averaged layers appear to have nearly constant rigidity. 
In view of~\citeauthor{Thomsen1986} parameters, in Table~\ref{tab:table1ass}(c), the average medium is nearly isotropic, as expected.

\begin{table}
\begin{minipage}[b]{0.3\linewidth}
\centering
\begin{tabular}{||l||l||}
\hline
\multicolumn{2}{||c||}{$[{\rm m}^2/{\rm s}^2]\times 10^6$}\\
\hline\hline  &\\[-10pt]
$c_{1111}^{\overline{\rm TI}}$ & 7.1903\\ \hline &\\[-10pt]
$c_{1133}^{\overline{\rm TI}}$ & 3.8508\\  \hline&\\[-10pt]
$c_{1212}^{\overline{\rm TI}}$ & 1.6698\\ \hline &\\[-10pt]
$c_{2323}^{\overline{\rm TI}}$ & 1.6698\\  \hline &\\[-10pt]
$c_{3333}^{\overline{\rm TI}}$ & 7.1904\\ 
\hline
\end{tabular}
\subcaption{}
\end{minipage}
\begin{minipage}[b]{0.3\linewidth}
\centering
\begin{tabular}{||l||l||}
\hline
\multicolumn{2}{||c||}{$[{\rm m}^2/{\rm s}^2]\times 10^6$}\\
\hline\hline
$\overline{c_{1111}}$ & 7.1915\\ \hline 
$\overline{c_{2323}}$ & 1.6698\\  \hline
${\rm std}_{1111}$ & 0.0923\\ \hline
${\rm std}_{2323}$ & 0.0059\\  \hline
\end{tabular}
\subcaption{}
\end{minipage}
\begin{minipage}[b]{0.3\linewidth}
\centering
\begin{tabular}{||l||l||}
\hline
\multicolumn{2}{||c||}{$\times 10^{-6}$}\\
\hline\hline
$\gamma$ & 5.5905\\ \hline
$\delta$ & -10.1212\\  \hline
$\epsilon$ & -6.1280\\ \hline
\end{tabular}
\subcaption{}
\end{minipage}
\caption{(a) Backus averages (b) means and standard deviations (c) Thomsen parameters,
for ten 0.1524-metre layers.}
\label{tab:table1ass}
\end{table}

%%%%%%%%%%%%%%%%%%%%%%%%%
\subsection{Transversely isotropic layers: Isotropic medium}\label{sec:translayers}
%%%%%%%%%%%%%%%%%%%%%%%%%
Even though, in general, transversely isotropic layers result---by the~\citeauthor{Backus1962} average---in a transversely isotropic medium, there exists a case for which inhomogeneity of the stack of transversely isotropic layers results in an isotropic medium.
Let us examine such a case.

\begin{lemma}
\label{lem:TIiso}
A transversely isotropic tensor with $c_{1111}=c_{3333}$\,, $c_{1133}=c_{1111}-2c_{2323}$\,, $c_{1212}\neq c_{2323}$ and $c_{2323}$ being constant is transversely isotropic.	
\end{lemma}
\begin{proof}
Consider
\begin{equation*}
C=
\left[
\begin{array}{cccccc}
c_{1111} & c_{1111}-2c_{1212} & c_{1111}-2c_{2323} & 0 & 0 & 0\\
c_{1111}-2c_{1212} & c_{1111} & c_{1111}-2c_{2323} & 0 & 0 & 0\\
c_{1111}-2c_{2323} & c_{1111}-2c_{2323} & c_{1111} & 0 & 0 & 0\\
0 & 0 & 0 & 2c_{2323} & 0 & 0\\
0 & 0 & 0 & 0 & 2c_{2323} & 0\\
0 & 0 & 0 & 0 & 0 & 2c_{1212}\end{array}
\right]
\,.
\end{equation*}
Its eigenvalues are
\begin{equation*}
\lambda_1=\tfrac{3}{2}c_{1111}-c_{1212}-\frac{\sqrt{9c_{1111}^2-32c_{1111}c_{2323}-4c_{1111}c_{1212}+32c_{2323}^2+4c_{1212}^2}}{2}\,,
\end{equation*}
\begin{equation*}
\lambda_2=\tfrac{3}{2}c_{1111}-c_{1212}+\frac{\sqrt{9c_{1111}^2-32c_{1111}c_{2323}-4c_{1111}c_{1212}+32c_{2323}^2+4c_{1212}^2}}{2}\,,
\end{equation*}
\begin{equation*}
\lambda_3=\lambda_4=2c_{2323}\,,
\end{equation*}
\begin{equation*}
\lambda_5=\lambda_6=2c_{1212}\,,
\end{equation*}
which---due to the eigenvalue multiplicities---implies that $C$ is a transversely isotropic tensor~\citep{BonaEtAl2007}, as required.
\end{proof}
\begin{proposition}
\label{prop:TIiso}	
The~\citeauthor{Backus1962} average of a stack of transversely isotropic layers with $c_{1111}=c_{3333}$\,, $c_{1133}=c_{1111}-2c_{2323}$\,, $c_{1212}\neq c_{2323}$ and $c_{2323}$ being constant (Lemma~\ref{lem:TIiso}), can result---depending on the values of parameters---in an isotropic medium.
\end{proposition}
\begin{proof}
In general, the~\citeauthor{Backus1962} average of transversely isotropic layers is~\citep[e.g.,][Section 4.2.3]{Slawinski2020b}
\begin{equation}
\label{eq:pola}
c^{\overline{\rm TI}}_{1111}=
\overline{\left(c_{1111}-\dfrac{c_{1133}^2}{c_{3333}}\right)}
+\overline{\left(\dfrac{c_{1133}}{c_{3333}}\right)}^{\,2}
\,\,\,\overline{\left(\dfrac{1}{c_{3333}}\right)}^{\,-1}
\,,
\end{equation}
\begin{equation}
\label{eq:takietam}
c^{\overline{\rm TI}}_{1133}=\overline{\left(\dfrac{c_{1133}}{c_{3333}}\right)}
\,\,\,\overline{\left(\dfrac{1}{c_{3333}}\right)}^{\,-1}
\,,
\end{equation}
\begin{equation}\label{eq:madzia}
c^{\overline{\rm TI}}_{1212}=\overline{c_{1212}}
\,,
\end{equation}
\begin{equation}\label{eq:milek}
c^{\overline{\rm TI}}_{2323}=\overline{\left(\dfrac{1}{c_{2323}}\right)}^{\,-1}
\,,
\end{equation}
\begin{equation}
\label{eq:adam}
c^{\overline{\rm TI}}_{3333}=\overline{\left(\dfrac{1}{c_{3333}}\right)}^{\,-1}
\,.
\end{equation}
Isotropy of the average requires
\begin{equation}\label{eq:cond1}
c^{\overline{\rm TI}}_{1212}=c^{\overline{\rm TI}}_{2323}\,,
\end{equation}
\begin{equation}\label{eq:cond2}
c^{\overline{\rm TI}}_{1111}=c^{\overline{\rm TI}}_{3333}\,,
\end{equation}
\begin{equation}\label{eq:cond3}
c^{\overline{\rm TI}}_{1133}=c^{\overline{\rm TI}}_{1111}-2c^{\overline{\rm TI}}_{2323}\,.
\end{equation}
To satisfy condition~(\ref{eq:cond1}), we equate relations~(\ref{eq:madzia}) and (\ref{eq:milek}).
Since $c_{2323}$ is constant,
\begin{equation*}
c^{\overline{\rm TI}}_{2323} 
=
\overline{\left(\dfrac{1}{c_{2323}}\right)}^{\,-1}
=
\overline{c_{2323}}
=
c_{2323}
=
\overline{c_{1212}}
=
c^{\overline{\rm TI}}_{1212}\,.
\end{equation*} 
To satisfy condition~(\ref{eq:cond2}), we equate relations~(\ref{eq:pola}) and (\ref{eq:adam}).
Since $c_{1111}=c_{3333}$\,, $c_{1133}=c_{1111}-2c_{2323}$\,,
\begin{align*}
c^{\overline{\rm TI}}_{1111}
&=
\overline{\left(c_{1111}-\dfrac{c_{1133}^2}{c_{3333}}\right)} + \overline{\left(\dfrac{c_{1133}}{c_{3333}}\right)}^{\,2}
\,\,\,\overline{\left(\dfrac{1}{c_{3333}}\right)}^{\,-1}
\\
&=
\overline{\left(\dfrac{c_{1111}-2c_{2323}}{c_{1111}}\right)}^{\,2}\,\,\,
\overline{\left(\dfrac{1}{c_{1111}}\right)}^{\,-1}
+\overline{\left(\dfrac{4(c_{1111}-c_{2323})c_{2323}}{c_{1111}}\right)}
\\
&=
\overline{\left(\dfrac{1}{c_{1111}}\right)}^{\,-1}=\overline{\left(\dfrac{1}{c_{3333}}\right)}^{\,-1}=c^{\overline{\rm TI}}_{3333}\,,
\end{align*}
as required.
To satisfy condition~(\ref{eq:cond3}), we equate relations~(\ref{eq:pola}), (\ref{eq:takietam}), (\ref{eq:milek}).
Since $c_{1111}=c_{3333}$\,, $c_{1133}=c_{1111}-2c_{2323}$ and $c_{2323}$ is constant,
\begin{align*}
c^{\overline{\rm TI}}_{1133}
&=
\overline{\left(\dfrac{c_{1133}}{c_{3333}}\right)}
\,\,\,\overline{\left(\dfrac{1}{c_{3333}}\right)}^{\,-1}
=
\overline{\left(\dfrac{c_{1111}-2c_{2323}}{c_{1111}}\right)}\,\,\,\,\overline{\left(\dfrac{1}{c_{1111}}\right)}^{\,-1}
\\
&=
\overline{\left(\dfrac{1}{c_{1111}}\right)}^{\,-1}-2c_{2323}
=
c^{\overline{\rm TI}}_{1111}-2c^{\overline{\rm TI}}_{2323}\,,
\end{align*}
as required, which completes the proof.
\end{proof}
%%%%%%%%%%%%%%%%%%%%%%%%%
\section{Conclusions}
\label{sec:Conclusions}
%%%%%%%%%%%%%%%%%%%%%%%%%
For a stack of isotropic layers, the strength of anisotropy---resulting from the \citeauthor{Backus1962} average---is solely a measure of inhomogeneity.
However, if $c_{2323}$ is constant, then that inhomogeneity of $c_{1111}$ alone does not result in anisotropy.
In other words, the anisotropy of the~\citeauthor{Backus1962} average is a consequence of the difference in rigidity among layers, not in compressibility.

A physical counterpart of such a mathematical model might be a porous rock of constant rigidity, whose compressibility varies depending on the amount of liquid within its pores.
Following such a physical interpretation, and according to the~\citeauthor{Backus1962} average, the level of saturation alone has no effect on the isotropy of the medium, even though it has an effect on the value of $c^{\overline{\rm TI}}_{1111}$\,, whose value determines the $P$-wave propagation speed.

It is impossible to distinguish---from the~\citeauthor{Backus1962} average---if the stack of isotropic layers is homogeneous in both elasticity parameters or homogeneous in $c_{2323}$ only. 
Let us consider a numerical example.

If $c_{1111}=10$ and $c_{2323}=2$\,, then---regardless of the number of layers---$c^{\overline{\rm TI}}_{1111}=10$\,, $c^{\overline{\rm TI}}_{1133}=6$\,, $c^{\overline{\rm TI}}_{1212}=2$\,, $c^{\overline{\rm TI}}_{2323}=2$\,, $c^{\overline{\rm TI}}_{3333}=10$\,; the average is isotropic.
For a case discussed in Section~\ref{sec:iso}, we let  $c_{1111}$: $20$\,, $10$\,, $20$\,, $5$\,, $20$\,, $20$\,, $5$\,, $5$\,, $20$\,, $20$\,, and we let $c_{2323}=2$\,, for all layers.
The~\citeauthor{Backus1962} average is the same as for $c_{1111}=10$ and $c_{2323}=2$\,.

Furthermore, as illustrated in Appendix~\ref{sec:inverseprob}, the~\citeauthor{Backus1962} average of transversely isotropic layers can again result in the same values of the isotropic elasticity parameters.
Thus, from the~\citeauthor{Backus1962} average that results in an isotropic medium, it is possible to infer neither the material symmetry of layers nor the constancy of $c_{1111}$\,.
%%%%%%%%%%%%%%%%%%%%%%%%%%%%
\section*{Acknowledgments}
%%%%%%%%%%%%%%%%%%%%%%%%%%%%
We wish to acknowledge the graphic support of Elena Patarini and computer support of Izabela Kudela.
This research was performed in the context of The Geomechanics Project supported by Husky Energy. 
Also, this research was partially supported by the Natural Sciences and Engineering Research Council of Canada, grant 238416-2013.
%%%%%%%%%%%%%%%%%%%%%%%%%
\bibliographystyle{apa}
\bibliography{AKSS.bib}
%%%%%%%%%%%%%%%%%%%%%%%%%
\begin{appendix}
\section{Transversely isotropic layers: special case}\label{sec:inverseprob}
\begin{table}[h]
\begin{tabular}{||p{0.1575\textwidth}||*{4}{p{0.1575\textwidth}||}} 
\hline
 $c_{1111}$ & $c_{1133}$ & $c_{1212}$ & $c_{2323}$ & $c_{3333}$\\
 \hline\hline
 20&16&3&2&20\\
 \hline
 5&1&1&2&5\\
 \hline
 20&16&3&2&20\\
 \hline
  20&16&1&2&20\\
 \hline
 20&16&2.5&2&20\\
 \hline
 20&16&1.5&2&20\\
 \hline
 5&1&1&2&5\\
 \hline
 10&6&1&2&10\\
 \hline
 5&1&3&2&5\\
 \hline
 20&16&3&2&20\\
 \hline
\end{tabular}
\caption{
Elasticity parameters of ten transversely isotropic layers
}
\label{tab:TI_parameters}
\end{table}

For the values in Table~\ref{tab:TI_parameters}, the~\citeauthor{Backus1962} average is
\begin{equation*}
c^{\overline{\rm TI}}_{1111}=10\,,\quad 
c^{\overline{\rm TI}}_{1133}=6\,,\quad 
c^{\overline{\rm TI}}_{1212}=2\,,\quad
c^{\overline{\rm TI}}_{2323}=2\,,\quad
c^{\overline{\rm TI}}_{3333}=10\,.
\end{equation*}	
\end{appendix}
%%%%%%%%%%%%%%%%%%%%%%%%%
\end{document}